\documentclass[9pt]{amsart}

\usepackage{amsmath,amsfonts,amssymb}
\usepackage[dvips]{graphicx}

\usepackage{array,delarray}
\usepackage{txfonts}

\usepackage{xspace}
\usepackage{color}

\newtheorem{theorem}{Theorem}[section]   % Numbered within each section
\newtheorem{proposition}[theorem]{Proposition}  % Numbered along with thm

\theoremstyle{definition}
\newtheorem{definition}[theorem]{Definition}   % Numbered along with thm

\theoremstyle{remark}
\newtheorem{example}[theorem]{Example}        % Numbered along with thm

\numberwithin{equation}{section}     % Number equations within sections
\begin{document}

\title[AND-NOT modeling]{AND-NOT logic framework for steady state analysis of Boolean network models}

\author[Veliz-Cuba]{Alan~Veliz-Cuba\lowercase{$^{ab}$}}
\author[Buschur]{Kristina Buschur\lowercase{$^b$}}
\author[Hamershock]{Rose Hamershock\lowercase{$^b$}}
\author[Kniss]{Ariel Kniss\lowercase{$^b$}}
\author[Wolff]{Esther Wolff\lowercase{$^b$}}
\author[Laubenbacher]{Reinhard Laubenbacher\lowercase{$^b$}\\ \\ \\
\lowercase{$^a$}University of Nebraska-Lincoln,
203 Avery Hall, Lincoln, NE 68588, USA.\\ Phone: 402-472-7233, Fax: 402-472-8466, aveliz-cuba2@unl.edu \\ \lowercase{$^b$}Virginia Bioinformatics Institute %, Blacksburg, VA, USA
}

\begin{abstract}
Finite dynamical systems (e.g. Boolean networks and logical models) have been used in modeling biological systems to focus attention on the qualitative features of the system, such as the wiring diagram. Since the analysis of such systems is hard, it is necessary to focus on subclasses that have the properties of being general enough for modeling and simple enough for theoretical analysis. In this paper we propose the class of AND-NOT networks for modeling biological systems and show that it provides several advantages. Some of the advantages include:
Any finite dynamical system can be written as an AND-NOT network with similar dynamical properties. There is a one-to-one correspondence between AND-NOT networks, their wiring diagrams, and their dynamics. Results about AND-NOT networks can be stated at the wiring diagram level without losing any information. Results about AND-NOT networks are applicable to any Boolean network. We apply our results to a Boolean model of Th-cell differentiation. 
\end{abstract}

\maketitle 
%%%%%%%%%%%%%%%%%%%%%%%
\section{Introduction}
%%%%%%%%%%%%%%%%%%%%%%%
Discrete models have a long and successful history in systems biology, beginning with 
Boolean network representations of molecular networks \cite{kauffman69b} and
their later generalization, so-called logical models \cite{TA}. They are qualitative, time-discrete models that
are particularly suitable for the analysis of steady state behavior of molecular networks. 
However, as models become larger it is increasingly difficult to analyze them. In order to keep  the analysis of such networks tractable, many studies have focused on specific classes of networks such as: single-switch, unate, nested canalizing, threshold, AND, AND-OR, and linear networks \cite{Wittmann2010436,El,CBN,ADG,sontag,Aracena_FB_BRN,rae,Kauff, Kauff2,NCF,Murrugarra,Murrugarra2}.  In order to be useful for modeling, a family of networks has to be ``sufficiently general'' for modeling biological interactions and ``simple enough'' for theoretical analysis. In this paper we propose the family of AND-NOT networks as such family. AND-NOT networks are a particular the class of Boolean networks that are constructed using only the AND ($\wedge$) and NOT ($\neg$) operators.

A biological justification for the use of AND-NOT networks is that there is evidence that for genes that are regulated by more than one other gene, the different
binding sites exhibit synergistic effects between the different regulators \cite{Nguyen,Gummow, merika}.
This fact motivated the study of conjunctive Boolean networks, that is, networks whose logical rules
are constructed using exclusively the AND operator  \cite{CBN},  where explicit formulas 
for steady states are given; also, upper and lower bounds for the number and length of limit cycles are provided. 
But conjunctive Boolean networks cannot account for inhibitory regulation and the resulting negative
feedback loops, which are common in gene regulatory networks. Allowing the NOT operator, in addition to the AND operator  (i.e. using AND-NOT networks), can make the family of networks sufficiently general to be useful for
modeling \cite{Park15062010}. 

For a formal argument that the family of AND-NOT networks is general enough for modeling, we will show that any discrete model (finite dynamical system, to be precise) can be represented by an AND-NOT network. More precisely, we present an algorithm that assigns to a given general
 discrete model an AND-NOT network which has the same number of steady states, together with an
algorithmic correspondence between steady states of the two networks. This is achieved by adding nodes
to the network as needed. The potential drawback of this algorithm is of course that the network size can potentially get significantly larger, thereby potentially negating any computational advantage gained by the specialized logic. 
However, since molecular networks have typically small in-degree, this growth in the number of network nodes to be added is
 modest in the case of molecular network models. We demonstrate this through an analysis of several
published models and random networks. 

To argue that AND-NOT networks are simple enough for theoretical analysis, we will show how using the specialized logic of AND-NOT networks can provide better theoretical results. For example, in \cite{Richard20093281}, it was shown that an upper bound for the number of steady states can easily be computed for AND-NOT networks (which is not true for arbitrary networks). Also, in \cite{SCBN}, it was shown that the exact number of steady states of AND-NOT networks are encoded in the topological features of the wiring diagram, and that, in some cases, the problem of finding the exact number of steady states can be transformed to the problem of finding maximal independent sets of the wiring diagram, which has been extensively studied \cite{Parallel,Eppstein,Gely20091447,Byskov2004547,Deterministic,lawler558,Makino,Schmidt2009417,Schneider,Wan}. In this paper we will show how the specialized logic of AND-NOT networks can give us better upper bounds for the number of steady states; more precisely, we provide an upper bound for AND-NOT networks that improves on previous upper bounds. Furthermore, we show how this upper bound for AND-NOT networks can actually be used for general networks. We use our results to analyze a Boolean model of Th-cell differentiation. Another theoretical advantage of AND-NOT networks is that they are in a one-to-one correspondence with their wiring diagrams. This observation has several implications, one of which is the possibility to relate dynamic network properties with features of the wiring diagram \cite{CBN,SCBN}. Also, from a given signed wiring diagram one can unambiguously construct and AND-NOT network, which implies that all algorithms or results can be stated at the ``wiring diagram level.'' 

%%%%%%%%%%%%%%%%%%%%%%%%%%%%%%%%%%%%%
\section{Definitions}
\label{sec-pre}
%%%%%%%%%%%%%%%%%%%%%%%%%%%%%%%%%%%%%

\begin{definition}
For a signed directed graph $G=(V_G,E_G)$, we denote $I_i=\{j:(j,i,s)\in E_G \}$,
 $I^+_i=\{j:(j,i,+)\in E_G \}$ and $I^-_i=\{j:(j,i,-)\in E_G \}$.% Similarly, $O_j=\{i:(j,i,s)\in E_G \}$, $O^+_j=\{i:(j,i,+)\in E_G \}$ and $O^-_j=\{i:(j,i,-)\in E_G \}$.
That is, $I_i$ is the set of all incoming edges for node $i$, and $I^+_i$, resp. $I^-_i$ is the subset of positive, resp. negative, edges. All graphs in the rest of the paper will be signed directed graphs unless noted otherwise.
\end{definition}

In order to simplify the graphical representation, we denote two negative (positive) edges between $i$ and $j$ by a bidirectional negative (positive) edge, $\multimapdotboth$ ($\blacktriangleleft\!$---$\!\blacktriangleright$).
%$\bullet\!$---$\!\bullet$   $\blacktriangleleft\!$---$\!\blacktriangleright$
If the edges have different signs we denote them by $\bullet\!$---$\!\blacktriangleright$.

\begin{definition}
An \emph{AND-NOT function} is a Boolean function, $h:\{0,1\}^n\rightarrow \{0,1\}$, such that $h$ can be written in the form $$h(x_1,\ldots,x_n)=\bigwedge_{j\in P}x_j \wedge \bigwedge_{j\in N}\neg
x_j,$$ where $P\cap N=\{ \ \}$. If $P=N=\{\ \}$, then $h$ is the constant function 1. If $i\in P$ ($i\in N$, respectively) we say that $i$ or $x_i$ is a \emph{positive} (\emph{negative}) regulator of $h$ or that it is an activator (repressor). An \emph{AND-NOT network} is a Boolean network (BN), $f=(f_1,\ldots,f_n):\{0,1\}^n\rightarrow\{0,1\}^n$, such that $f_i$ is an AND-NOT function for all $i=1,\ldots,n$. AND-NOT networks are also called \emph{signed conjunctive networks}.
\end{definition}

\begin{definition}
The \emph{wiring diagram} of an AND-NOT network is defined by a graph $G=(V_G,E_G)$ with vertices $V_G=\{1,\ldots,n\}$ (or $\{x_1,\ldots,x_n\}$) and edges $E_G$ given as follows: $(i,j,+)\in E_G$ ($(i,j,-)\in E_G$, respectively) if $x_i$ is a positive (negative, respectively) regulator of $f_j$. Notice that nodes corresponding to constant functions have in-degree zero. Also, the wiring diagram of an AND-NOT network contains all the information about the network; that is, we only need to specify the wiring diagram in order to define an AND-NOT network.
\end{definition}

\begin{example}\label{eg:ANBN} Consider the Boolean network $f=(f_1,\ldots,f_6):\{0,1\}^6\rightarrow\{0,1\}^6$ given by \\$f(x)=(x_2\wedge x_4\wedge \neg x_5, x_1\wedge x_6\wedge\neg x_3 \wedge\neg x_5,1,x_6\wedge \neg x_1\wedge\neg x_5,x_6\wedge \neg x_1,1)$. It is easy to see that $f$ is an AND-NOT network. Its wiring diagram is shown in Figure \ref{fig:ANBN}.

\end{example}

\begin{figure}[here]
\centerline{\framebox{\includegraphics[totalheight=2cm]{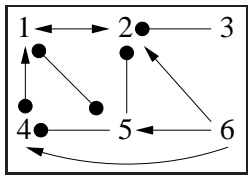}}}
\caption{Wiring diagram of the AND-NOT network in Example \ref{eg:ANBN}.}
\label{fig:ANBN}
\end{figure}

As mentioned in the introduction, some other families of networks that have been studied in the past are single-switch, linear, AND, AND-OR, unate and nested canalyzing functions \cite{Wittmann2010436,El,CBN,ADG,sontag,Aracena_FB_BRN,rae,NCF}. Each family has its own advantages; however, for the purpose of modeling biological systems and for theoretical analysis, it is of interest to have the following properties: First, networks generated using these families should be able to admit a sign assignment; that is, it should be possible to determine the sign of an interaction. Second, in principle, it should be general enough to model all networks; that is, it should be possible to model any type of regulation. Third, for theoretical analysis, it would be useful to have a one-to-one correspondence between wiring diagrams and networks. This property would allow complete encoding of a network in its wiring diagram. The family of linear functions satisfies the third property but not the first two. The family of AND functions satisfies the first and third property but not the second. The family of AND-OR functions satisfies the first property but not the last two. Single-switch, unate, and nested canalyzing functions  satisfy the first two properties but not the third.

On the other hand, AND-NOT networks satisfy all three properties. The first property is satisfied because the sign of a regulation is given by the presence or absence of the NOT operator. The third property follows from the fact that if the positive and negative edges to $i$ are given by $P$ and $N$, resp., then the function for node $i$ is uniquely given by $f_i=\bigwedge_{j\in P} x_j \wedge \bigwedge_{j\in N} \neg x_j$. The second property is given by the fact that any finite dynamical system can be expressed as an AND-NOT network. More precisely, Theorem \ref{thm:main} guarantees that steady states are preserved if we rewrite a general finite dynamical system as an AND-NOT network.

%%%%%%%%%%%%%%%%%%%%%%%%%%%%%%%%%%%%%
\section{Results}
\label{sec-res}
%%%%%%%%%%%%%%%%%%%%%%%%%%%%%%%%%%%%%

In this section we show why AND-NOT networks are a good framework for modeling biological systems.

%%%%%%%%%%%%%%%%%%%%%%%%%%%%%%%%%%%%%
\subsection{AND-NOT networks are general enough for modeling}
\label{sec-general}
%%%%%%%%%%%%%%%%%%%%%%%%%%%%%%%%%%%%%

One issue that can potentially arise when only using certain classes of networks is that one can have difficulty in modeling certain processes. For example, the family of AND networks does not allow modeling negative interactions. Another example is that the family of linear networks, does not allow modeling signed interactions. In order for a family of networks to be useful for modeling, is has to allow modeling any type of interaction. 

Here we show that the family of AND-NOT networks is general enough for modeling. More precisely, we show that for any finite dynamical system, there exists an AND-NOT network (possibly with more nodes) such that they share  key dynamical properties.

\begin{theorem}\label{thm:main}
Let $h=(h_1,\ldots,h_n):S\rightarrow S$ be a finite dynamical system, where $S=X_1\times \cdots \times X_n$ and all $X_i$'s are finite. Then, there exists an AND-NOT network $g:\{0,1\}^m\rightarrow \{0,1\}^m$ such that there is a bijection between the steady states of $h$ and $g$. Furthermore, $g$ and the bijection between steady states is given algorithmically. We say that $g$ is an AND-NOT representation of $h$.
\end{theorem}
\begin{proof}
 A simple proof uses the facts that any finite dynamical system can be written as a Boolean network \cite{Didier2011177}, and that any Boolean function has a conjunctive normal form. 
 
In \cite{Didier2011177}, the authors proved algorithmically that for any finite dynamical system $h$, there exists a Boolean network $f$ (possibly with more nodes) such that $h$ and $f$ have the same number of steady states. Furthermore, the bijection of steady states is also given algorithmically. Therefore, we only need to show that there exists and AND-NOT network $g$, such that there is a bijection between the steady states of $f$ and $g$.

We proceed by induction.
First, consider the conjunctive normal form of $f_n$: $f_n=w_1\wedge w_2\wedge\cdots\wedge w_r$, 
where $w_j$ is of the form $w_j(x)=s_1 x_1\vee s_2 x_2\vee\cdots \vee s_u x_u$ with $s_i\in\{id,\neg\}$ ($id=$identity function). Notice that $\neg w_j$ is an AND-NOT function. Then, define the BN $k=(k_1,\ldots,k_{n+r}):\{0,1\}^{n+r}\rightarrow \{0,1\}^{n+r}$ in variables $(x_1,\ldots,x_n,y_1,\ldots,y_r)$ by
$k_i(x,y)=f_i(x)$ for $i=1,\ldots,n-1$, $k_n(x,y)=\neg y_1\wedge \neg y_2\wedge \cdots \wedge \neg y_r$ and 
$k_i(x,y)=\neg w_i(x)$ for $i=n+1,\ldots,n+r$. 

We now check that the function $\phi(x)=(x,\neg w_1(x),\ldots,\neg w_r(x))$ gives a one-to-one correspondence between steady states of $f$ and $k$. Suppose that $f(x)=x$, then\\ $k(\phi(x))=k(f_1(x),\ldots,f_{n-1}(x), w_1(x)\wedge  w_2(x)\wedge \ldots \wedge  w_r(x),w_1(x),\ldots,w_r(x))$\\ $=k(f_1(x),\ldots,f_{n-1}(x), f_n(x),w_1(x),\ldots,w_r(x))=(x,\neg w_1(x),\ldots,\neg w_r(x))=\phi(x)$; that is, $\phi(x)$ is a steady state $k$. Now, suppose that $k(x,y)=(x,y)$ and notice that in this case $y_i=k_i(x,y)=\neg w_i(x)$; then $(x,y)=\phi(x)$. Also, $f(x)=(f_1(x),\cdots,f_{n-1}(x),f_n(x))$ $=(k_1(x,y),\ldots,k_{n-1}(x,y),w_1(x)\wedge w_2(x)\wedge\cdots\wedge w_r(x))$ $=(x_1,\ldots,x_{n-1},\neg y_1\wedge \neg y_2\wedge \cdots \wedge \neg y_r)$ $=(x_1,\ldots,x_{n-1},k_n(x,y))$ $=(x_1,\ldots,x_{n-1},x_n)=x$. That is, $x$ is a steady state of $f$. Therefore, $k=(k_1,\ldots,k_{n+r})$ is a BN where $k_n,\ldots,k_{n+r}$ are AND-NOT functions and such that there is a one-to-one correspondence between the steady states of $f$ and $k$. By induction, it follows that there is an AND-NOT network $g:\{0,1\}^m\rightarrow \{0,1\}^m$ together with a bijection between the steady states 
of $f$ and $g$. 

Therefore, there is a bijection between the steady states of $h$ and $g$. Furthermore, $g$ and the bijection are given algorithmically.
\end{proof}

The transformation of finite dynamical systems to Boolean networks has been discussed in \cite{Didier2011177}. So, in the rest of the paper we will focus on Boolean networks and AND-NOT networks.

\begin{example}\label{eg:basic}
Consider the BN $f:\{0,1\}^5\rightarrow \{0,1\}^5$ given by $f_1=x_2\vee \neg x_4$, $f_2=x_1\wedge x_3$, $f_3=(x_2\vee \neg x_4)\wedge x_5$, $f_4=x_3\vee x_5$, $f_5=x_3$. The wiring diagram of $f$ is given in Figure \ref{fig:basic} (left).
In order to transform this BN to an AND-NOT network we introduce the variable $x_6$ with Boolean function $f_6=\neg x_2 \wedge x_4$ and $f_7=\neg x_3 \wedge \neg x_5$. Variables $x_6$ and $x_7$ will be used in $g_1$ and $g_4$. Notice that since $x_2 \vee \neg x_4$ appears again in $f_3$, we can simply reuse $x_6$ to keep the number of extra variables as small as possible. Then the AND-NOT network is $g:\{0,1\}^7\rightarrow \{0,1\}^7$ given by $g_1=\neg x_6$, $g_2=x_1\wedge x_3$, $g_3=\neg x_6 \wedge x_5$, $g_4=\neg x_7$, $g_5=x_3$, $g_6=\neg x_2 \wedge x_4$, $g_7=\neg x_3 \wedge \neg x_5$. The wiring diagram of $g$ is shown in Figure \ref{fig:basic} (right).
\end{example}

\begin{figure}[here]
\centerline{\framebox{\includegraphics[totalheight=1.5cm]{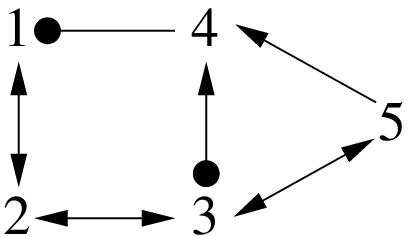}} 
\framebox{\includegraphics[totalheight=1.5cm]{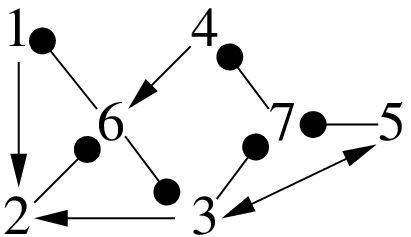}}}
\caption{Wiring diagram of the BN network $f$ and the
AND-NOT network $g$ in Example \ref{eg:basic}.}
\label{fig:basic}
\end{figure}

An additional step in the transformation that can keep the number of extra variables small is given by the following proposition.

\begin{proposition} \label{prop:ANDOR}
Let $f:\{0,1\}^n\rightarrow \{0,1\}^n$ be a BN and define $g:\{0,1\}^n\rightarrow \{0,1\}^n$ by
$g=N_k\circ f\circ N_k$, where $N_k(x_1,\ldots,x_n)=(x_1,\ldots,x_{k-1},\neg x_k,x_{k+1},\ldots,x_n)$. Then $f$ and $g$ are dynamically equivalent.
\end{proposition}
\begin{proof}
It is enough to notice that $N_k$ is invertible with inverse $N_k$. Then, $g^r=N_k\circ f^r \circ N_k$; that is, evaluating $f$ is equivalent to evaluating $g$.
\end{proof}

If some functions of a BN are OR-NOT functions, then we can use Proposition \ref{prop:ANDOR} to transform the BN into a BN in the same number of variables such that the OR-NOT functions become AND-NOT functions. Also, Proposition \ref{prop:ANDOR} can be used to transform constant functions $f_k=0$ into constant functions $f_k=1$ (if $f_k=0$, then the $k$-th coordinate function of $N_k\circ f\circ N_k$ is the constant function 1).

\begin{example}\label{eg:basicor}
Consider the BN $f:\{0,1\}^3\rightarrow \{0,1\}^3$ given by $f_1=x_2$, $f_2=x_1\vee \neg x_3$, $f_3=x_2\wedge x_3$. The wiring diagram of $f$ is in Figure \ref{fig:basicor} (left). Since $f_2$ is an OR-NOT function, we can transform it to a AND-NOT function using Proposition \ref{prop:ANDOR}.
Consider $g=N_2\circ f\circ N_2$, given by $g(x)=N_2(f(x_1,\neg x_2,x_3))$ $=N_2(\neg x_2,x_1\vee \neg x_3,\neg x_2 \wedge x_3)$ $=(\neg x_2,\neg(x_1\vee \neg x_3),\neg x_2 \wedge x_3)$ $=(\neg x_2,\neg x_1\wedge x_3,\neg x_2 \wedge x_3)$, with wiring diagram shown in Figure \ref{fig:basicor} (right). Then, $f$ is dynamically equivalent to an AND-NOT network. Notice that the effect of this transformation on the wiring diagram is simple, we simply change the signs of the edges around node 2.
\end{example}

\begin{figure}[here]
\centerline{\framebox{\includegraphics[totalheight=1.5cm]{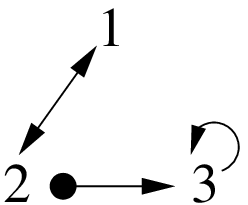}} 
\framebox{\includegraphics[totalheight=1.5cm]{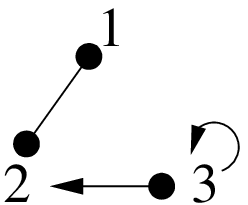}}}
\caption{Wiring diagram of the BN network $f$ and the
AND-NOT network $g$ in Example \ref{eg:basicor}.}
\label{fig:basicor}
\end{figure}

As mentioned in \cite{Didier2011177}, an advantage of transforming finite dynamical systems into Boolean networks is that it can provide insight into the role of feedback loops by disentangling them. In this sense, transforming finite dynamical systems into AND-NOT networks can pass all the information of the role of feedback loops to the wiring diagram. In this case, the wiring diagram is not only a rough representation of the network, but it encodes all the information of the network; in this sense the wiring diagram ``becomes'' the network. This has the potential to reduce the problem of studying the structure of the state space graph (which has $2^n$ elements) to studying the structure of the wiring diagram of the AND-NOT representation (which has $m\geq n$ elements). This can help in understanding the precise role of the network topology in the network dynamics. A similar approach was used successfully to study conjunctive and linear networks \cite{CBN,El}.

%%%%%%%%%%%%%%%%%%%%%%%%%%%%%%%%%%%%%
\subsection{The variable growth in AND-NOT representation is small}
\label{sec-growth}
%%%%%%%%%%%%%%%%%%%%%%%%%%%%%%%%%%%%%

For practical purposes it is important to obtain an estimate of how much the 
AND-NOT representation can increase the number of variables.  For arbitrary Boolean networks, the number of extra nodes can be exponential in the number of nodes. However, Boolean models of biological systems are not arbitrary and are actually very sparse with very low in-degree (typically described by a power law distribution \cite{Albert01112005,Huynen01051998}). We will now show that in practice the number of variables introduced by the algorithm can be small.

\begin{table}[here]
\caption{Number of extra variables introduced by the AND-NOT representation. The number of nodes of $f$ and its AND-NOT representation, $g$, are denoted by $n$, $m$, respectively. The BNs were taken from \cite{Remy,Velizlacop,mendozamethod,erbb2,Klamt_th}.
}\label{table:num}
\begin{tabular}{|l|l|r|}
  \hline
   $n$ & $m$ & \% increase \\
  \hline
   12 &   13  & 8\% \\
  \hline
   12 &   15  & 20\%\\
   \hline
   14  &  15   & 7\% \\
   \hline
   20 & 24  & 20\%    \\
   \hline
    23  &  26   & 13\% \\
   \hline
    28  &  28   & 0\% \\
  \hline
    40  & 43 & 7.5\% \\
  \hline
\end{tabular}
\end{table}

\begin{table}[here]
\caption{Average number of extra variables introduced by the AND-NOT representation for random BNs.}\label{table:k}
\begin{tabular}{|l|r|}
  \hline
   in-deg$\leq K$ & \% increase \\
  \hline
   $K=1, 2$ &  0\% \\
  \hline
   $K=3$ &  5.2\% \\
  \hline
   $K=4$ &   10.8\%\\
  \hline
     $K=5$ &   16.2\%\\
  \hline
     $K=6$ &   20.8\%\\
  \hline
       $K=7$ &    24.8\%\\
  \hline
       $K=8$ &     28.6\%\\
  \hline
       $K=9$ &      32.3\%\\
  \hline
       $K=10$ &     36.1\%\\
  \hline
\end{tabular}
\end{table}

In order to study this question, we have applied the procedure to several published models in the literature and studied the question using randomly generated Boolean networks. The first study shows that the increase in the number of variables
for published models is modest (Table \ref{table:num}). The number of variables was increased by 14\% on average with a maximum value of 4 extra nodes. In order to determine the number of extra nodes introduced by our algorithm for more general BNs, we did a statistical analysis. To mimic wiring diagrams coming from biological systems, the edges followed a power law distribution and we considered the maximum in-degree less than or equal to $K$ for $K=1,\ldots,10$ (see Appendix A for details). The results of this second study are shown in Table \ref{table:k}. For example, all networks where nodes have in-degree bounded by $K=2$ can be transformed to AND-NOT networks without increasing the number of nodes. For networks where nodes have in-degree bounded by $K=4$, our method increases the number of nodes by $10.8\%$ on average (see Appendix A for details). 
It is important to mention that in both tables, the growth in the number of extra nodes is far less than exponential.

%%%%%%%%%%%%%%%%%%%%%%%%%%%%%%%%%%%%%
\subsection{AND-NOT networks can be useful is theoretical analysis}
\label{sec-theory}
%%%%%%%%%%%%%%%%%%%%%%%%%%%%%%%%%%%%%

As mentioned in the Introduction, the specialized logic of AND-NOT networks can be used to obtain better theoretical results. Such results can arise directly (e.g. \cite{CBN,SCBN}) or by applying results about general Boolean networks to the family of AND-NOT networks.  In this section we show examples of the latter. First, we need the following definitions. 

Let $C$ be a feedback loop of a graph $G$. We say that $C$ is a \emph{strong } feedback loop if there are no edges of the form $k\rightarrow i$,$k\multimapdot j$ in $G\setminus C$ such that $i,j\in C$.

For example, consider the graph $G$ in Figure \ref{fig:strong}. The feedback loop $\{3,4\}$ is not strong because of the edges $1\rightarrow 3$, $1\multimapdot 4$; $\{5,6\}$ and $\{1,2,4,3,5\}$ are not strong because of the edges $1\rightarrow 5$, $1\multimapdot 4$. All other feedback loops are strong.

\begin{figure}[here]
\centerline{\framebox{\includegraphics[totalheight=2.5cm]{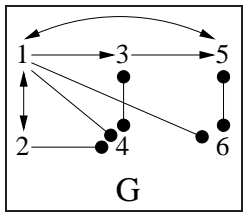}}}
\caption{Graph with only one strong positive feedback loop.}
\label{fig:strong}
\end{figure}

Our first result in this section is an application of \cite[Theorem 3.2]{Remy2008335} to the family of AND-NOT networks (see Appendix B for the proof).

\begin{theorem}\label{thm:strong}
Let $W$ be the wiring diagram of an AND-NOT network, and suppose $J$ intersects all strong positive feedback loops of $W$. Then, the number of steady states is at most $2^{|J|}$.
\end{theorem}

\begin{example}
Consider the AND-NOT network with wiring diagram given in Figure \ref{fig:strong}. The only strong positive feedback loops are $\{1,2\}$ and $\{1,3,5\}$. Since $J=\{1\}$ intersects them, Theorem \ref{thm:strong} guarantees that there are at most $2^{|J|}=2$ steady states.
\end{example}

Intuitively, Theorem \ref{thm:strong} is telling us which positive feedback loops contribute to the presence of steady states; it says that they have to be strong. We also provide a slight generalization of Theorem \ref{thm:strong}. We need the following definition.

A feedback loop $C$ of a graph $W$ is called \textit{inconsistent} if there is a vertex $k_C$ such that there is a positive path of the form $k_C\rightarrow i_1\rightarrow \cdots \rightarrow i_r \rightarrow t_C$ from $k_C$ to $t_C\in C$ and a negative path of the form $k_C\rightarrow j_1\rightarrow \cdots \rightarrow j_r \multimapdot u_C$, from $k_C$ to $u_C\in C$ such that $k_C\rightarrow t_C$, $k_C\multimapdot u_C$ are not edges in $C$ and $|I_{j_1}|=\ldots=|I_{j_r}|=1$. When such vertex $k_C$ does not exist, we say that $C$ is \emph{consistent}.	

For example, consider the graph $W$ in Figure \ref{fig:completion}. The positive feedback loop $\{3,4\}$ is inconsistent because of the paths $1\rightarrow 3$ and $1\rightarrow 2\multimapdot 4$. The positive feedback loop $\{5,6\}$ is inconsistent because of the paths $1\rightarrow 3\rightarrow 5$ and $1\multimapdot 6$. Also, the positive feedback loop $\{1,2,4,3,5\}$ is inconsistent because of the paths $1\rightarrow 3\rightarrow 5$ and $1\rightarrow 2\multimapdot 4$. Then, the only consistent feedback loops are $\{1,2\}$ and $\{1,3,5\}$.

\begin{figure}[here]
\centerline{\framebox{\includegraphics[totalheight=2.5cm]{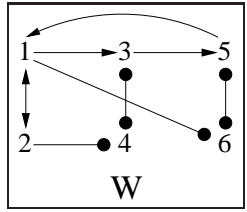}} }
\caption{Wiring diagram of the AND-NOT network in Example \ref{eg:completion}. }
\label{fig:completion}
\end{figure}

We say that a set $J\subseteq \{1,\ldots,n\}$ \textit{dominates} a graph $W$ if $J$ intersects all consistent positive feedback loop and for each feedback loop $C$ that is inconsistent and strong, $J$ intersects $C$ or $J$ contains at least one $k_C$. For example, the set $\{1\}$ dominates the graph $W$ in Figure \ref{fig:completion}.

With these definitions we have the following theorem that gives an upper bound on the number of steady states using topological features of the wiring diagram (see Appendix B for the proof).

\begin{theorem}\label{thm:badv}
Let $W$ be the wiring diagram of an AND-NOT network, and suppose $J$ dominates $W$. Then the number of steady states is at most $2^{|J|}$.
\end{theorem}

It is not difficult to see that the bound given by Theorem \ref{thm:strong} is greater than or equal than the bound given by Theorem \ref{thm:badv}. The next example shows that the inequality is in some cases strict.

\begin{example}\label{eg:completion} Consider the BN $f:\{0,1\}^6\rightarrow \{0,1\}^6$ given by\\
$\begin{array}{lll}
f_1 & =& x_2\wedge x_5,  \\
f_2 & =&  x_1, \\
f_3 & =&  x_1\wedge \neg x_4,\\
f_4 & =&  \neg x_2\wedge \neg x_3, \\
f_5 & =&  x_3\wedge \neg x_6, \\
f_6 & =&  \neg x_1 \wedge \neg x_5,
\end{array}$
\end{example}

Its wiring diagram is shown in Figure \ref{fig:completion}. It is easy to see that $\{1,3,5\}$ intersects all strong positive feedback loops. Then, Theorem  \ref{thm:strong} gives the upper bound $2^3=8$. On the other hand, since $\{1\}$ dominates the wiring diagram, Theorem \ref{thm:badv} gives the upper bound 2. That is, Theorem \ref{thm:badv} gave a better upper bound on the number of steady states. Notice that in this case the actual number of steady states is 2, namely, $0 0 0 1 0 1$ and $1 1 1 0 1 0$.

One might argue that having better results for AND-NOT networks is not enough to justify their use. After all, since we are considering a smaller family of Boolean networks we should of course obtain stronger results. However, the combination of Theorem \ref{thm:main} and results about AND-NOT networks automatically generates theorems for all Boolean networks. Furthermore, such combination can in some cases provide stronger results. This deserves further explanation which is illustrated in Figure \ref{fig:idea}. Consider a theorem about Boolean networks that gives us information about certain dynamical properties, ``Thm.''. On the other hand, consider a similar theorem about AND-NOT networks,  ``Thm.$^*$''. Then, given a Boolean network $f$, we have two choices, we can apply  Thm. to $f$; or, we can use Theorem \ref{thm:main} to find the AND-NOT representation of $f$, then apply $\textrm{Thm.}^*$, and then use Theorem \ref{thm:main} to obtain information about the original Boolean network $f$. In Section \ref{sec-bio} we use a published Boolean model to show that the latter can  give stronger results.

\begin{figure}[here]
\centerline{\framebox{\includegraphics[totalheight=2.5cm]{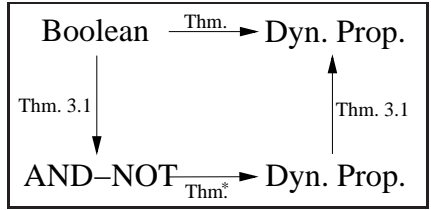}}  }
\caption{Extension of theorems about AND-NOT network to all Boolean networks.}
\label{fig:idea}
\end{figure}

For example, combining Theorem \ref{thm:main} and \ref{thm:badv} we obtain the following theorem.

\begin{theorem}\label{thm:BNbound}
Let $f$ be any Boolean network and suppose that $J$ dominates the wiring diagram of its AND-NOT representation. Then, $f$ has at most $2^{|J|}$ steady states.
\end{theorem}

We now show that this theorem can in fact provide a better upper bound for the number of steady states.

%%%%%%%%%%%%%%%%%%%%%%%%%%%%%%%%%%%%%
\subsection{Application to Th-cell differentiation}
\label{sec-bio}
%%%%%%%%%%%%%%%%%%%%%%%%%%%%%%%%%%%%%\\
We apply our results to the BN model proposed in \cite{mendozamethod} for Th-cell differentiation. The model is a BN in 23 variables, $f:\{0,1\}^{23}\rightarrow \{0,1\}^{23}$. 
Below is the list of Boolean functions. The wiring diagram is shown in Figure \ref{fig:Th}.

$\begin{array}{lllllll}
x_1 &=& GATA3     & ,  & f_1	 & =& (x_1\vee x_{21})\wedge \neg x_{22}; \\
x_2 &=& IFN-\beta  & ,  & f_2	 & =& 0; \\
x_3 &=& IFN-\beta R & ,  & f_3   & =& x_2;\\
x_4 &=& IFN-\gamma & ,   & f_4	 & =&  (x_{14}\vee x_{16}\vee x_{20}\vee x_{22})\wedge \neg x_{19};\\
x_5 &=& IFN-\gamma R & ,   & f_5  & =&  x_4;\\
x_6 &=& IL-10 & ,   & f_6  & =&  x_1;\\
x_7 &=& IL-10R & ,   & f_7  & =& x_6; \\
x_8 &=& IL-12 & ,   & f_8  & =& 0; \\
x_9 &=& IL-12R & ,   & f_9  & =& x_8\wedge \neg x_{21};\\
x_{10} &=& IL-18 & ,   & f_{10}   & =& 0;\\
x_{11} &=& IL-18R & ,   & f_{11}   & =& x_{10}\wedge \neg x_{21};\\
x_{12} &=& IL-4 & ,   & f_{12}   & =& x_1\wedge \neg x_{18};\\
x_{13} &=& IL-4R & ,   & f_{13} & =& x_{12}\wedge \neg x_{17};\\
x_{14} &=& IRAK & ,   & f_{14}  & =& x_{11};\\
x_{15} &=& JAK1 & ,   & f_{15}   & =&  x_5\wedge\neg x_{17};\\
x_{16} &=& NFAT & ,   & f_{16} & =&  x_{23};\\
x_{17} &=& SOCS1 & ,   & f_{17}  & =& x_{18}\vee x_{22};\\
x_{18} &=& STAT1 & ,   & f_{18}  & =& x_3\vee x_{15};\\
x_{19} &=& STAT3 & ,   & f_{19}  & =& x_7;\\
x_{20} &=& STAT4 & ,   & f_{20}  & =& x_9\wedge \neg x_1;\\
x_{21} &=& STAT6 & ,   & f_{21}   & =& x_{13};\\
x_{22} &=& T-bet & ,   & f_{22}   & =& (x_{18}\vee x_{22})\wedge\neg x_1;\\
x_{23} &=& TCR & ,   & f_{23}   & =&   0.
\end{array}$

\begin{figure}[here]
\centerline{\framebox{\includegraphics[totalheight=6cm]{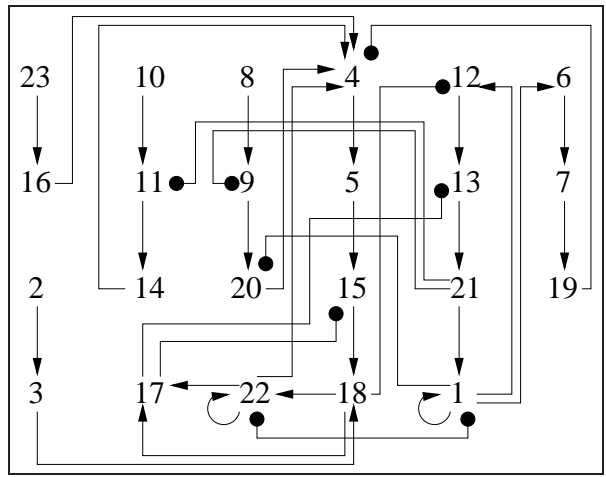}} }
\caption{Wiring diagram of the Th-cell differentiation model.}
\label{fig:Th}
\end{figure}

Using our algorithms we obtain the AND-NOT network, $g:\{0,1\}^{26}\rightarrow \{0,1\}^{26}$, shown in Figure \ref{fig:ThAN}. It turns out that the set $\{1,22\}$ dominates the wiring diagram of $g$ (see Appendix C for details). Then, by Theorem \ref{thm:BNbound}, the number of steady states of $f$ is at most $2^2=4$. On the other hand, all previous results about steady states (e.g. \cite{Aracena_FB_BRN,Richard20093281}) give 8 as the upper bound. That is, using the AND-NOT representation can provide a better upper bound, even for general Boolean networks. The actual number of steady states of the model is 3 (see \cite{mendozamethod} for details).

\begin{figure}[here]
\centerline{\framebox{\includegraphics[totalheight=6cm]{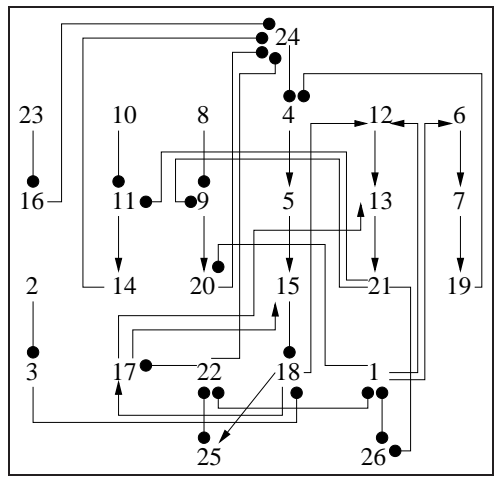}} }
\caption{Wiring diagram of the AND-NOT representation of the Th-cell differentiation model.}
\label{fig:ThAN}
\end{figure}

%%%%%%%%%%%%%%%%%%%%%%%%%%%%%%%%%%%%%
\section{Discussion}
\label{sec-dis}
%%%%%%%%%%%%%%%%%%%%%%%%%%%%%%%%%%%%%
The results presented in this paper, together with other results in the literature, support that the family of AND-NOT networks are general enough for modeling and simple enough for theoretical analysis. Given any finite dynamical system, it is possible to create an AND-NOT network such that they have similar dynamical properties. This has two implications: First, this means that using AND-NOT networks in modeling does not pose any technical restriction on the type of interactions one can model. Second, every result about AND-NOT networks can be applied to general Boolean networks, which can give better results (e.g. Theorem \ref{thm:BNbound}). One potential drawback for this framework is that the AND-NOT representation can have more nodes. However, for networks that arise from modeling biological systems, this increase in the number of nodes is modest (Section \ref{sec-growth}). 

Other advantages of using AND-NOT networks are the following: First, all information about the network is actually contained in the network's wiring diagram. Specifically, there is a one-to-one correspondence between AND-NOT networks and graphs, so that the network can be reconstructed unambiguously from the wiring diagram. In \cite{RuiSheng} the authors followed a similar approach to successfully study cascading effects. Second, due to this correspondence, we can state all results about AND-NOT networks using wiring diagrams only. This means that questions about AND-NOT networks can be reformulated as questions about graphs; then, one can use tools from graph theory and combinatorics to study them (e.g. antichains, posets, inclusion-exclusion principle, independent sets \cite{CBN,SCBN}). This deserves further investigation.

Finally, we point out that AND-NOT networks are special cases of so-called \emph{nested canalyzing} Boolean networks. These were first introduced in \cite{Kauff, Kauff2} as good candidates for models with ``biologically meaningful" regulatory rules, and have since been studied extensively. In \cite{Murrugarra} this concept was generalized to multi-state models, and it was shown there that the large majority of regulatory rules that appear in published models of biological networks are of this form. It was shown furthermore that nested canalyzing networks have dynamic properties one would expect to find in biological networks, such as short limit cycles and a small number of attractors. Thus, the results in the present paper imply that in order to study the steady state behavior of general network models, one can focus on the very restrictive class of nested canalyzing networks \cite{Murrugarra2}, instantiated as AND-NOT networks and make use of their very special properties.

%%%%%%%%%%%%%%%%%%%%%%%%%%%%%%%%%%%%%
\section*{Appendix A}
%%%%%%%%%%%%%%%%%%%%%%%%%%%%%%%%%%%%%
We describe here the details of the study to determine how many nodes are added by
the construction of the AND-NOT representation. To mimic wiring diagrams coming from biological systems, the edges followed a power law distribution. More precisely, given $K$ fixed and a parameter $\gamma$, the probability for a node to have $k\leq K$ nodes is $p_k=k^{-\gamma}$ (up to a normalization factor). For example, if $K=4$, the probabilities of having $1$, $2$, $3$ and 4 nodes are $p_1=c1^{-\gamma}=c$, $p_2=c2^{-\gamma}$, $p_3=c3^{-\gamma}$ and $p_4=c4^{-\gamma}$, respectively, where $c=\frac{1}{1^{-\gamma}+2^{-\gamma}+3^{-\gamma}+4^{-\gamma}}$ so that $p_1+p_2+p_3+p_4=1$. Also, to mimic biological regulation, we restricted our analysis to Boolean functions that admitted a sign assignment for the edges. These Boolean functions are called unate, biologically meaningful and regulatory functions \cite{sontag,rae,Aracena_FB_BRN}.

Denote with $e_k$ the average number of extra nodes introduced by a Boolean function in $k$ variables. Then, a BN that follows the distribution mentioned above will have, on average, $p_1e_1+p_2e_2+\cdots+p_Ke_K$ extra nodes. Now, we need to estimate $e_k$. 

Consider a Boolean function, $h$, that depends on $k$ variables. For $k=1$ there are 2 functions, $h=x_1$ and $h=\neg x_1$ and we do not need to introduce any new nodes; then $e_1=0$. For $k=2$ there are 8 functions and they are of the form $h=s_1x_1\wedge s_2x_2$ or $h=s_1x_1\vee s_2x_2$, where $s_ix_i=x_i$ or $s_ix_i=\neg x_i$. For functions of the form $h=s_1x_1\wedge s_2x_2$ we do not introduce any new nodes, and for functions of the form $h=s_1x_1\vee s_2x_2$ we can use Proposition \ref{prop:ANDOR} to transform $h$ to an AND-NOT function, so we do not introduce new nodes either. Then $e_2=0$. For $k=3$, there are 72 functions. An exhaustive-search analysis shows that of those 72 Boolean functions, 16 introduce 0 nodes, 48 introduce 1 node, and 8 introduce 3 nodes; then the average number of extra nodes in this case is $e_3=\frac{16*0+48*1+8*3}{72}=1$. For $k=4$, there are 1824 Boolean functions. An exhaustive-search analysis shows that of those 1824 functions, 32 introduce 0 nodes, 320 introduce 1 node, 480 introduce 2 nodes, 960 introduce 3 nodes and 32 introduce 4 nodes; thus the average number of extra nodes in this case is $e_4=\frac{32*0+320*1+480*2+960*3+32*4}{1824}=2.35$. For $k=5$, there are 220608 functions and an exhaustive-search analysis shows that $e_5=4.03$. For $k=6$ there are approximately 
$5\times 10^8$ functions and an exhaustive-search analysis would be unfeasible. However, we have the following result.

%\begin{theorem}
\textbf{Theorem A.1.}
\emph{
The average number of extra nodes for a unate function of $k$ variables is at most $C(k,\lfloor k/2 \rfloor)$; that is, $e_k\leq C(k,\lfloor k/2 \rfloor)$. Where $C$ is the binomial coefficient and $\lfloor \ \rfloor$ is the floor function.
}
%\end{theorem}
\begin{proof}
Without loss of generality we assume the CNF of the Boolean function $f$ has no negative signs. Let $f=w_1\wedge \ldots \wedge w_r$ be the CNF, where $w_i$ has the form $w_i=x_1\vee \ldots \vee x_s$. For each $i$, define $S_i=\{l:x_l \textrm{ appears in } w_i\}$. 

Now, if there are $i$, $j$ such that $S_i\subseteq S_j$, then we can simplify $w_i\wedge w_j$ to $w_i$ (e.g. $(x_1\vee x_2)\wedge (x_1\vee x_2\vee x_3)=x_1\vee x_2$). That is, we can simplify the CNF so that $S_i\nsubseteq S_j$ for all $i\neq j$. 

Thus, $S_1,\ldots,S_r$ is a family of subsets of $\{1,\ldots,k\}$ such that no one is contained in the other. Sperner's theorem \cite{Sperner} states that $r \leq C(k,\lfloor k/2 \rfloor)$. This implies that for any unate function in $k$ variables, we need at most $\leq C(k,\lfloor k/2 \rfloor)$ extra nodes to obtain the AND-NOT representation. Therefore, $e_k\leq C(k,\lfloor k/2 \rfloor)$.

\end{proof}

It is important to mention that the exhaustive-search analysis done for $k=3,4,5$ suggests that $e_k$ is actually much smaller than $C(k,\lfloor k/2 \rfloor)$. In fact, we did a statistical analysis for $k=6,\ldots,10$ using a total of 5000000 Boolean functions chosen at random (1000000 for each $k$). The analysis shows the following approximations: $e_6\approx 5.32$, $e_7\approx 7.04$, $e_8\approx 9.32$,  $e_9\approx 12.24$, $e_{10}\approx 15.96$.

Table \ref{table:k} shows a summary of our analysis for $\gamma=2.5$. For example, if $K=4$, then the fractions of functions with 1, 2, 3 and 4 variables are on average $p_1=.786$, $p_2=.139$, $p_3=.0504$ and $p_4=.0246$, respectively. Then, the average number of extra nodes is: 
$$100(p_1e_1+p_2e_2+p_3e_3+p_4e_4 )=100(.786*0+.139*0+.0504*1 +.0246*2.35)\approx 10.8\%.$$

%%%%%%%%%%%%%%%%%%%%%%%%%%%%%%%%%%%%%
\section*{Appendix B}
%%%%%%%%%%%%%%%%%%%%%%%%%%%%%%%%%%%%%

Here we prove Theorem \ref{thm:strong} and \ref{thm:badv}. As mentioned in Section \ref{sec-theory}, Theorem \ref{thm:strong} is an application of  \cite[Theorem 3.2]{Remy2008335} to the family of AND-NOT networks. First we need the following definition.

Let $f:\{0,1\}^n\rightarrow \{0,1\}^n$ be a Boolean network and consider $x\in \{0,1\}^n$. 
Then, $W(x)=(V,E)$ is the graph with vertices $V=\{1,\ldots,n\}$ and the following edges:\\
$(j,i,+)\in E$ if $x_j=0$ and $f_i(x)<f_i(x+e_j)$, or if $x_j=1$ and $f_i(x-e_j)<f_i(x)$;\\
$(j,i,-)\in E$ if $x_j=0$ and $f_i(x)>f_i(x+e_j)$, or if $x_j=1$ and $f_i(x-e_j)>f_i(x)$;\\
where $e_j$ is the vector given by $(e_j)_i=\delta_{ij}$ ($\delta$ is the Kronecker delta).
Notice that if $(j,i,+)$ or $(j,i,-)$ is an edge in $W(x)$, then changing the $j$-th coordinate of $j$ produces a change in $f_j$.
Notice that for AND-NOT networks we have that $W(x)\subseteq W$ for all $x$; in fact, this is true for more general networks.

%\begin{lemma}
\textbf{Theorem B.1.\cite{Remy2008335}}
\emph{
Let $f$ be a Boolean network and suppose $a$ and $b$ are steady states of $f$. Then, there there exists $x$ such that $W(x)$ has a positive feedback loop with vertices in the set $\{i:a_i\neq b_i\}$.\\
}

We now prove Theorem \ref{thm:strong}.
\begin{proof}
Let $\phi:\{0,1\}^n\rightarrow \{0,1\}^{|J|}$ defined by $\phi(x)=x_J$. 
We will show that if $a\neq b$ are steady states of $g$, then $\phi(a)\neq \phi(b)$. Consider $a\neq b$ steady states of $g$; then, by Theorem B.1., there exists $x$ such that $W(x)$ has a positive feedback loop, $C$, with vertices in the set $\{i:a_i\neq b_i\}$. 

We claim that $C$ is a strong positive feedback loop of $W$. By contradiction, suppose there is $k\in\{1,\ldots,n\}$ and $i,j\in C$ such that $k\rightarrow i$ and $k\multimapdot j$ are edges in $W(x)$ but not in $C$. Then, $W(x)$ has edges of the form $(l_1,i,\pm)$ and $(l_2,j,\pm)$ where $l_1,l_2\neq k$.
On the other hand, since $(k,i,+),(k,j,-)\in C\subseteq W(x)\subseteq W$, we have that $f_i=x_k\wedge \ldots$ and $f_j=\neg x_k \wedge \ldots$. We have two cases $x_k=0$ or $x_k=1$. In the case $x_k=0$ we obtain that $f_i=0$ for all values of $x_1,\ldots,x_{k-1},x_{k+1},\ldots,x_n$. In particular, $W(x)$ cannot have an edge of the form $(l,i,\pm)$ with $l\neq k$; this is a contradiction. In the case $x_k=1$ we obtain that $f_j=0$ for all values of $x_1,\ldots,x_{k-1},x_{k+1},\ldots,x_n$. In particular, $W(x)$ cannot have an edge of the form $(l,j,\pm)$ with $l\neq k$; this is a contradiction as well. 
Therefore,  $C$ is strong.

Since $C$ is a strong positive feedback loop in $W$, $C$ must intersect $J$. Since $C$ has all its vertices in the set $\{i:a_i\neq b_i\}$, $J$ intersects the set $\{i:a_i\neq b_i\}$. Therefore $\phi(a)=a_J\neq b_J= \phi(b)$. It follows that the restriction of $\phi$ to the set of steady states is an injective function. Therefore, $|\{x:f(x)=x\}|\leq | \{0,1\}^{|J|} |=2^{|J|}$.
\end{proof}

It is important to mention that Theorem \ref{thm:strong} was also proven in \cite{Richard20093281} using different techniques.\\

We now prove Theorem \ref{thm:badv}. 
\begin{proof}
Let $f:\{0,1\}^n\rightarrow \{0,1\}^n$ be an AND-NOT network with wiring diagram $W$. Let $C$ be a positive feedback loop that is strong and inconsistent. Then, there is a vertex $k_C$ such that there is a positive path of the form $k_C\rightarrow i_1\rightarrow \cdots \rightarrow i_r \rightarrow t_C$ from $k_C$ to $t_C\in C$ and a negative path of the form $k_C\rightarrow j_1\rightarrow \cdots \rightarrow j_r \multimapdot u_C$, from $k_C$ to $u_C\in C$ such that $k_C\rightarrow t_C$, $k_C\multimapdot u_C$ are not edges in $C$ and $|I_{j_1}|=\ldots=|I_{j_r}|=1$. Let $G$ be the graph obtained by adding to $W$ all edges of the form $k_C\rightarrow t_C$ and $k_C\multimapdot u_C$ where $C$ does not intersect $J$. Denote by $g:\{0,1\}^n\rightarrow \{0,1\}^n$ the AND-NOT network associated to $Z$. We claim that the steady states of $f$ and $g$ are the same. We prove this by induction on the number of extra edges. 

Suppose that $W$ and $Z$ only differ in the edge $k\rightarrow t$, then, by definition we must also have a path $k\rightarrow i_1\rightarrow \ldots \rightarrow i_r \rightarrow t$. Suppose that $g(x)=x$, we need to show that $f_j(x)=x_j$ for all $j$.
Since $W$ and $Z$ only differ in the edge $k\rightarrow t$ we have $f_j=g_j$ for $j\neq t$, $g_t=f_t\wedge x_k$ and $f_t=x_{i_r}\wedge \ldots$. Then, $f_j(x)=g_j(x)=x_j$ for $j\neq t$. It remains to show that $f_t(x)=x_t$. Consider first the case $x_t=0$, then, $g_t(x)=0$ and $x_i=0$ for some $i\in I_t^+$. If $i\neq k$, we have that the edge $i\rightarrow t$ is in $W$ and $f_t=x_{i_r}\wedge x_i\wedge \ldots$; then, $f_t(x)=x_{i_r}\wedge 0 \wedge\ldots=0=x_t$. If $i=k$, then $x_k=0$ which implies that $x_{i_1}=0$ (because of the edge $k\rightarrow i_1$); similarly, we obtain that $x_{i_r}=0$. Then, $f_t(x)=0\wedge \ldots =0=x_t$. That is, $f_t(x)=x_t$. Now consider the case $x_t=1$. Since $1=x_t=g_t(x)=f_t(x)\wedge x_k$, we have $f_t(x)=1=x_t$. A similar argument shows that if $f(x)=x$, then $g(x)=x$. The proof for when $W$ and $Z$ only differ in the edge $k\multimapdot t$ is analogous. By induction we obtain that $f$ and the AND-NOT network obtained by a completion of $W$ have the same steady states.

Now, we claim that $J$ intersects all strong positive feedback loops of $Z$. Let $C'$ be a strong positive feedback loop of $Z$. Then we have two cases: $C'$ is in $W$ or it is not. Consider the case $C'\subseteq W$. Then, $C'$ is a strong positive feedback loop in $W$. If $C'$ is consistent in $W$, then it intersects $J$. If $C'$ is inconsistent (and strong) in $W$, then it also intersects $J$. Now consider the case $C'\nsubseteq W$. Then, at least one edge of $C'$ is of the form $k_C\rightarrow t_C$ or $k_C\multimapdot u_C$ for some $C$ strong and inconsistent that does not intersect $J$. Then, $k_C\in J$ and $J$ intersects $C$. In any case we obtain that $J$ intersects all strong positive feedback loops of $Z$.

Then, the number of steady states of $g$, and hence $f$, is at most $2^{|J|}$.
\end{proof}

%%%%%%%%%%%%%%%%%%%%%%%%%%%%%%%%%%%%%
\section*{Appendix C}
%%%%%%%%%%%%%%%%%%%%%%%%%%%%%%%%%%%%%

We first analyze the original BN using previous results. In \cite{mendozamethod}, the authors showed that the positive feedback loops of the BN $f:\{0,1\}^{23}\rightarrow f:\{0,1\}^{23}$ are:

$\{4,5,15,18,12,13,21,11,14\}$

$\{4,5,15,18,12,13,21,9,20\}$

$\{4,5,15,18,12,13,21,1,6,7,19\}$

$\{4,5,15,18,12,13,21,1,20\}$

$\{4,5,15,18,12,13,21,1,22\}$

$\{4,5,15,18,17,13,21,11,14\}$

$\{4,5,15,18,17,13,21,9,20\}$

$\{4,5,15,18,17,13,21,1,6,7,19\}$

$\{4,5,15,18,17,13,21,1,20\}$

$\{4,5,15,18,17,13,21,1,22\}$

$\{4,5,15,18,22\}$

$\{4,5,15,18,22,17,13,21,11,14\}$

$\{4,5,15,18,22,17,13,21,9,20\}$

$\{4,5,15,18,22,17,13,21,1,6,7,19\}$

$\{4,5,15,18,22,17,13,21,1,20\}$

$\{4,5,15,18,22,1,12,13,21,11,14\}$

$\{4,5,15,18,22,1,12,13,21,9,20\}$

$\{4,5,15,18,22,1,6,7,19\}$

$\{4,5,15,18,22,1,20\}$

$\{12,13,21,1\}$

$\{13,21,1, 22,17\}$

$\{22\}$

$\{22,1\}$

$\{1\}$

We will use the following two theorems (proven in \cite{Aracena_FB_BRN,Richard20093281}, respectively) that give upper bounds on the number of steady states.

\begin{theorem}\label{thm:pfv2}
Let $W$ be the wiring diagram of a BN network and suppose $J$ is a set of vertices that intersects all positive  feedback loops in $W$. Then, the number of steady states is at most $2^{|J|}$.
\end{theorem}

\begin{theorem}\label{thm:fpfv}
Let $W$ be the wiring diagram of a BN network and suppose $J$ is a set of vertices that intersects all functional positive  feedback loops in $W$. Then, the number of steady states is at most $2^{|J|}$.
\end{theorem}

It is easy to see that all positive feedback loops intersect the set $\{1,4,22\}$. Therefore, Theorem \ref{thm:pfv2} gives the upper bound $2^3=8$. Also, it is possible to show that the functional positive feedback loops are $\{4,5,15,18,12,13,21,11,14\}$, $\{22\}$, $\{22,1\}$ and $\{1\}$ (e.g. using the GINsim software \cite{ginsim} ). Therefore, Theorem \ref{thm:fpfv} gives the upper bound 8 as well.

We now analyze the AND-NOT network using our results. The positive feedback loops of the AND-NOT network in Figure \ref{fig:ThAN} are the following (new nodes are in bold).

$\{\textbf{24},4,5,15,18,12,13,21,11,14\}$

$\{\textbf{24},4,5,15,18,12,13,21,9,20\}$

$\{\textbf{24},4,5,15,18,12,13,21,\textbf{26},1,6,7,19\}$

$\{\textbf{24},4,5,15,18,12,13,21,\textbf{26},1,20\}$

$\{\textbf{24},4,5,15,18,12,13,21,\textbf{26},1,22\}$

$\{\textbf{24},4,5,15,18,17,13,21,11,14\}$

$\{\textbf{24},4,5,15,18,17,13,21,9,20\}$

$\{\textbf{24},4,5,15,18,17,13,21,\textbf{26},1,6,7,19\}$

$\{\textbf{24},4,5,15,18,17,13,21,\textbf{26},1,20\}$

$\{\textbf{24},4,5,15,18,17,13,21,\textbf{26},1,22\}$

$\{\textbf{24},4,5,15,18,\textbf{25},22\}$

$\{\textbf{24},4,5,15,18,\textbf{25},22,17,13,21,11,14\}$

$\{\textbf{24},4,5,15,18,\textbf{25},22,17,13,21,9,20\}$

$\{\textbf{24},4,5,15,18,\textbf{25},22,17,13,21,\textbf{26},1,6,7,19\}$

$\{\textbf{24},4,5,15,18,\textbf{25},22,17,13,21,\textbf{26},1,20\}$

$\{\textbf{24},4,5,15,18,\textbf{25},22,1,12,13,21,11,14\}$

$\{\textbf{24},4,5,15,18,\textbf{25},22,1,12,13,21,9,20\}$

$\{\textbf{24},4,5,15,18,\textbf{25},22,1,6,7,19\}$

$\{\textbf{24},4,5,15,18,\textbf{25},22,1,20\}$

$\{12,13,21,1\}$

$\{13,21,\textbf{26},1, 22,17\}$

$\{22,\textbf{25}\}$

$\{22,1\}$

$\{1,\textbf{26}\}$

Those feedback loops that contain 4 and 13 are inconsistent because of the paths $1\rightarrow 12\rightarrow 13$, $1\rightarrow 6\rightarrow 7\rightarrow 19\multimapdot 4$; they are also strong. All other positive feedback loops are consistent and intersect $\{1,22\}$. That is, $\{1,22\}$ intersects all consistent positive feedback loops, and for each positive feedback loop $C$ that is inconsistent and strong, $J$ contains $k_C=1$. Hence, $\{1,22\}$ dominates the wiring diagram of $g$. Therefore, Theorem \ref{thm:BNbound} gives the better upper bound  $2^2=4$ on the number of steady states of $f$.

\section*{Acknowledgement}
The research was funded by NSF grants CMMI-0908201 and
DMS-1062878.


\begin{thebibliography}{10}

\bibitem{kauffman69b}
S.~Kauffman, ``Homeostasis and differentiation in random genetic control
  networks,'' {\em Nature}, vol.~224, pp.~177--178, 1969.

\bibitem{TA}
R.~Thomas and R.~D'Ari, {\em Biological Feedback}.
\newblock Boca Raton, FL: CRC Press, 1990.

\bibitem{Wittmann2010436}
D.~M. Wittmann, C.~Marr, and F.~J. Theis, ``Biologically meaningful update
  rules increase the critical connectivity of generalized kauffman networks,''
  {\em Journal of Theoretical Biology}, vol.~266, no.~3, pp.~436 -- 448, 2010.

\bibitem{El}
B.~Elspas, ``The theory of autonomous linear sequential networks,'' {\em IRE
  Transaction on Circuit Theory}, pp.~45--60, 1959.

\bibitem{CBN}
A.~Jarrah, R.~Laubenbacher, and A.~Veliz-Cuba, ``The dynamics of conjunctive
  and disjunctive {B}oolean network models,'' {\em Bull. Math. Bio.}, vol.~72,
  no.~6, pp.~1425--1447, 2010.

\bibitem{ADG}
J.~Aracena, J.~Demongeot, and E.~Goles, ``Fixed points and maximal independent
  sets in {AND}-{OR} networks,'' {\em Discrete Appl. Math.}, vol.~138, no.~3,
  pp.~277--288, 2004.

\bibitem{sontag}
E.~Sontag, A.~Veliz-Cuba, R.~Laubenbacher, and A.~Jarrah, ``The effect of
  negative feedback loops on the dynamics of {B}oolean networks,'' {\em
  Biophysical Journal}, vol.~95, pp.~518--526, 2008.

\bibitem{Aracena_FB_BRN}
J.~Aracena, ``Maximum number of fixed points in regulatory {B}oolean
  networks,'' {\em Bulletin of Mathematical Biology}, vol.~70, no.~5,
  pp.~1398--1409, 2008.

\bibitem{rae}
L.~Raeymaekers, ``Dynamics of {B}oolean networks controlled by biologically
  meaningful functions,'' {\em J. Theor. Biol.}, vol.~218, no.~3, pp.~331--341,
  2002.

\bibitem{Kauff}
S.~Kauffman, C.~Peterson, B.~Samuelsson, and C.~Troein, ``Genetic networks with
  canalyzing {B}oolean rules are always stable,'' {\em PNAS}, vol.~101, no.~49,
  pp.~17102--17107, 2004.

\bibitem{Kauff2}
S.~Kauffman, C.~Peterson, B.~Samuelsson, and C.~Troein, ``{Random {B}oolean
  network models and the yeast transcriptional network},'' {\em PNAS},
  vol.~100, no.~25, pp.~14796--14799, 2003.

\bibitem{NCF}
A.~Jarrah, B.~Raposa, and R.~Laubenbacher, ``Nested canalyzing, unate cascade,
  and polynomial functions,'' {\em Physica D:Nonlinear Phenomena}, vol.~233,
  no.~2, pp.~167--174, 2007.

\bibitem{Murrugarra}
D.~Murrugarra and R.~Laubenbacher, ``Regulatory patterns in molecular
  interaction networks,'' {\em J. Theor. Biol.}, vol.~288, pp.~66--72, 2011.

\bibitem{Murrugarra2}
D.~Murrugarra and R.~Laubenbacher, ``The number of multistate nested canalyzing
  functions,'' {\em Physica D}, accepted, 2012.

\bibitem{Nguyen}
D.~H. Nguyen and P.~D'haeseleer, ``Deciphering principles of transcription
  regulation in eucaryotic genomes,'' {\em Mol. Sys. Biol.},
  no.~doi:10.1038/msb4100054, 2006.

\bibitem{Gummow}
B.~Gummow, J.~Sheys, V.~Cancelli, and G.~Hammer, ``Reciprocal regulation of a
  glucocorticoid receptor-steroidogenic factor-1 transcription complex on the
  dax-1 promoter by glucocorticoids and adrenocorticotropic hormone in the
  adrenal cortex,'' {\em Mol. Endocrinology}, vol.~20, no.~11, pp.~2711--2723,
  2006.

\bibitem{merika}
M.~Merika and S.~Orkin, ``Functional synergy and physical interactions of the
  erythroid transcription factor gata-1 with the kr\"uppel family proteins sp1
  and eklf,'' {\em Mol. Cell. Biol.}, vol.~15, no.~5, pp.~2437--2447, 1995.

\bibitem{Park15062010}
I.~Park, K.~Lee, and D.~Lee, ``{Inference of combinatorial Boolean rules of
  synergistic gene sets from cancer microarray datasets},'' {\em
  Bioinformatics}, vol.~26, no.~12, pp.~1506--1512, 2010.

\bibitem{Richard20093281}
A.~Richard, ``Positive circuits and maximal number of fixed points in discrete
  dynamical systems,'' {\em Discrete Applied Mathematics}, vol.~157, no.~15,
  pp.~3281 -- 3288, 2009.

\bibitem{SCBN}
A.~Veliz-Cuba and R.~Laubenbacher, ``On the computation of fixed points in
  {B}oolean networks,'' {\em Journal of Applied Mathematics and Computing},
  accepted, 2011.

\bibitem{Parallel}
N.~Du, B.~Wu, L.~Xu, B.~Wang, and P.~Xin, ``Parallel algorithm for enumerating
  maximal cliques in complex network,'' in {\em Mining Complex Data}
  (D.~Zighed, S.~Tsumoto, Z.~Ras, and H.~Hacid, eds.), vol.~165 of {\em Studies
  in Computational Intelligence}, pp.~207--221, Berlin / Heidelberg: Springer,
  2009.

\bibitem{Eppstein}
D.~Eppstein, ``All maximal independent sets and dynamic dominance for sparse
  graphs,'' {\em ACM Trans. Algorithms}, vol.~5, pp.~38:1--38:14, November
  2009.

\bibitem{Gely20091447}
A.~G\'ely, L.~Nourine, and B.~Sadi, ``Enumeration aspects of maximal cliques
  and bicliques,'' {\em Discrete Applied Mathematics}, vol.~157, no.~7,
  pp.~1447 -- 1459, 2009.

\bibitem{Byskov2004547}
M.~Jesper, ``Enumerating maximal independent sets with applications to graph
  colouring,'' {\em Operations Research Letters}, vol.~32, no.~6, pp.~547 --
  556, 2004.

\bibitem{Deterministic}
F.~Kuhn, T.~Moscibroda, T.~Nieberg, and R.~Wattenhofer, ``Fast deterministic
  distributed maximal independent set computation on growth-bounded graphs,''
  in {\em Distributed Computing} (P.~Fraigniaud, ed.), vol.~3724 of {\em
  Lecture Notes in Computer Science}, pp.~273--287, Berlin / Heidelberg:
  Springer, 2005.

\bibitem{lawler558}
E.~Lawler, J.~Lenstra, and A.~R. Kan, ``Generating all maximal independent
  sets: Np-hardness and polynomial-time algorithms,'' {\em SIAM Journal on
  Computing}, vol.~9, no.~3, pp.~558--565, 1980.

\bibitem{Makino}
K.~Makino and T.~Uno, ``New algorithms for enumerating all maximal cliques,''
  in {\em Algorithm Theory - SWAT 2004} (T.~Hagerup and J.~Katajainen, eds.),
  vol.~3111 of {\em Lecture Notes in Computer Science}, pp.~260--272, Berlin /
  Heidelberg: Springer, 2004.

\bibitem{Schmidt2009417}
M.~Schmidt, N.~Samatova, K.~Thomas, and B.~Park, ``A scalable, parallel
  algorithm for maximal clique enumeration,'' {\em Journal of Parallel and
  Distributed Computing}, vol.~69, no.~4, pp.~417 -- 428, 2009.

\bibitem{Schneider}
J.~Schneider and R.~Wattenhofer, ``A log-star distributed maximal independent
  set algorithm for growth-bounded graphs,'' in {\em Proceedings of the
  twenty-seventh ACM symposium on Principles of distributed computing}, PODC
  '08, (New York, NY, USA), pp.~35--44, ACM, 2008.

\bibitem{Wan}
L.~Wan, B.~Wu, N.~Du, Q.~Ye, and P.~Chen, ``A new algorithm for enumerating all
  maximal cliques in complex network,'' in {\em Advanced Data Mining and
  Applications} (X.~Li, O.~Zaiiane, and Z.~Li, eds.), vol.~4093 of {\em Lecture
  Notes in Computer Science}, pp.~606--617, Berlin / Heidelberg: Springer,
  2006.

\bibitem{Didier2011177}
G.~Didier, E.~Remy, and C.~Chaouiya, ``Mapping multivalued onto {B}oolean
  dynamics,'' {\em Journal of Theoretical Biology}, vol.~270, no.~1, pp.~177 --
  184, 2011.

\bibitem{Albert01112005}
R.~Albert, ``Scale-free networks in cell biology,'' {\em Journal of Cell
  Science}, vol.~118, no.~21, pp.~4947--4957, 2005.

\bibitem{Huynen01051998}
M.~Huynen and E.~van Nimwegen, ``The frequency distribution of gene family
  sizes in complete genomes.,'' {\em Molecular Biology and Evolution}, vol.~15,
  no.~5, pp.~583--589, 1998.

\bibitem{Remy}
E.~Remy, P.~Ruet, L.~Mendoza, D.~Thieffry, and C.~Chaouiya, ``From logical
  regulatory graphs to standard {P}etri nets: Dynamical roles and functionality
  of feedback circuits,'' {\em In Transactions on Computation Systems Biology
  VII (TCSB)}, pp.~55--72, 2006.

\bibitem{Velizlacop}
A.~Veliz-Cuba and B.~Stigler, ``Boolean models can explain bistability in the
  \textit{lac} operon,'' {\em J. Comput. Biol.}, vol.~18, no.~6, pp.~783--794,
  2011.

\bibitem{mendozamethod}
L.~Mendoza and I.~Xenarios, ``A method for the generation of standardized
  qualitative dynamical systems of regulatory networks,'' {\em Theoretical
  Biology and Medical Modelling}, vol.~3, no.~1, p.~13, 2006.

\bibitem{erbb2}
O.~Sahin, H.~Frohlich, C.~Lobke, U.~Korf, S.~Burmester, M.~Majety, J.~Mattern,
  I.~Schupp, C.~Chaouiya, D.~Thieffry, A.~Poustka, S.~Wiemann, T.~Beissbarth,
  and D.~Arlt, ``Modeling erbb receptor-regulated g1/s transition to find novel
  targets for de novo trastuzumab resistance,'' {\em BMC Systems Biology},
  vol.~3, no.~1, p.~1, 2009.

\bibitem{Klamt_th}
S.~Klamt, J.~Saez-Rodriguez, J.~Lindquist, L.~Simeoni, and E.~Gilles, ``A
  methodology for the structural and functional analysis of signaling and
  regulatory networks.,'' {\em BMC Bioinformatics}, vol.~7, no.~56, 2006.

\bibitem{Remy2008335}
E.~Remy, P.~Ruet, and D.~Thieffry, ``Graphic requirements for multistability
  and attractive cycles in a boolean dynamical framework,'' {\em Advances in
  Applied Mathematics}, vol.~41, no.~3, pp.~335 -- 350, 2008.

\bibitem{RuiSheng}
R.~Wang and R.~Albert, ``Elementary signaling modes predict the essentiality of
  signal transduction network components,'' {\em BMC Systems Biology}, vol.~5,
  no.~1, p.~44, 2011.

\bibitem{Sperner}
E.~Sperner, ``Ein satz {\"u}ber untermengen einer endlichen menge,'' {\em
  Mathematische Zeitschrift}, vol.~27, pp.~544--548, 1928.
\newblock 10.1007/BF01171114.

\bibitem{ginsim}
A.~Gonzalez, A.~Naldi, L.~S\'{a}nchez, D.Thieffry, and C.~Chaouiya, ``{GIN}sim
  : a software suite for the qualitative modelling, simulation and analysis of
  regulatory networks,'' {\em Biosystems}, vol.~84, no.~2, pp.~91--100, 2006.

\end{thebibliography}
\end{document}